%% file: arxiv.tex
\documentclass{article}
\usepackage{spconf,amsmath,graphicx}

\usepackage{amssymb,amsthm}
\usepackage{mathtools}
\usepackage{url}
\usepackage{xcolor}
\usepackage{verbatim}
\usepackage{subfig}
\usepackage{hyperref}

\input{ICASSP/definitions}

\makeatletter
\def\blfootnote{\gdef\@thefnmark{}\@footnotetext}
\makeatother

\title{A Statistical Interpretation of the Maximum Subarray Problem}
\twoauthors
  {Dennis Wei*}
	{IBM Research\\
	Yorktown Heights, NY, USA}
  {Dmitry M. Malioutov*}
	{\\
	Scarsdale, NY, USA}
\begin{document}
\ninept
\maketitle
\blfootnote{\copyright 2023 IEEE.  Personal use of this material is permitted.  Permission from IEEE must be obtained for all other uses, in any current or future media, including reprinting/republishing this material for advertising or promotional purposes, creating new collective works, for resale or redistribution to servers or lists, or reuse of any copyrighted component of this work in other works.}
\begin{abstract}
Maximum subarray is a classical problem in computer science that given an array of numbers aims to find a contiguous subarray with the largest sum. We focus on its use for a noisy statistical problem of localizing an interval with a mean different from background. While a naive application of maximum subarray fails at this task, both a penalized and a constrained version can succeed. We show that the penalized version can be derived for common exponential family distributions, in a manner similar to the change-point detection literature, and we interpret the resulting optimal penalty value. The failure of the naive formulation is then explained by an analysis of the estimated interval boundaries. Experiments further quantify the effect of deviating from the optimal penalty. We also relate the penalized and constrained formulations and show that the solutions to the former lie on the convex hull of the solutions to the latter.
\end{abstract}
\begin{keywords}
Maximum subarray, maximum-sum segment, change-point localization, Lagrangian relaxation, exponential family
\end{keywords}

\section{Introduction}
\label{sec:intro}

Given a one-dimensional (1D) array (i.e., a sequence) of numbers, the \emph{maximum (sum) subarray problem}
is to find a contiguous subarray with the largest sum \cite{bentley84}. Beyond its use in computer science, the maximum subarray problem has applications and generalizations in various domains. In computational biology, it is used for DNA sequence analysis (e.g.~to find GC-rich areas or DNA-binding domains) 
and has been generalized to include length constraints \cite{lin2002efficient} and alternative objectives such as maximum density \cite{max_density_segment_Lu}. In image processing and computer vision, a 2D version of maximum subarray can locate rectangles that differ the most in brightness or some other characteristic from the rest of the image, for example for astronomical images \cite{weddell2013maximum}, and can also be used to count objects \cite{lempitsky2010learning}. Parallel versions and specialized hardware implementations have been developed to accelerate the 2D maximum subarray problem \cite{takaoka_kadane2D_parallel, kadane2D_GPU}.

In this work, we consider the use of maximum subarray for a statistical problem of localizing an interval (or a rectangle in higher dimensions) with mean different than the background in high noise. We first find in Section~\ref{sec:prob} that a naive application of maximum subarray fails at this task, even when the mean difference is clearly larger than the noise level. The failure can be corrected either by solving a penalized version, in which a penalty value is subtracted from each element in the array, or a constrained version, in which the length of the estimated interval is bounded. In the paper we justify and relate these corrected formulations.

In Section~\ref{sec:stat_form}, we show that the penalized problem arises for array elements $w_t$ drawn from exponential family distributions, with the mild condition that one of the sufficient statistics is $w_t$ itself.
These include common distributions such as the Gaussian, Poisson, Negative-Binomial and Gamma. The optimal penalty is shown to be an intermediate value between the interval mean and the background mean, where the exact value depends on the distribution.

In Section~\ref{sec:alt_loc}, we provide an analysis of the estimated interval boundaries and localization error resulting from solving the penalized maximum subarray problem. 
In the case of naive maximum subarray with no penalty and zero-mean background, we show that the expected localization error is on the order of the length of the array, thus explaining the failure seen in Section~\ref{sec:prob}. 

In Section~\ref{sec:lagrange}, we relate the penalized maximum subarray problem to the constrained version, specifically to its Lagrangean form.  We observe that the solutions to the penalized problem are a subset of the solutions to the constrained one, lying at the vertices of the convex hull of the latter.  While we do not prove it formally in the paper, we suggest how this convex hull phenomenon can arise based on integer programming duality.

Section~\ref{sec:expt} presents additional experiments on localization performance as a function of the mean difference and the penalty value used, to quantify the effect of deviating from the optimal penalty. We add an illustrative example of 2D max-subarray used to localize vehicles in a synthetic aperture radar (SAR) image.

\section{The Maximum Subarray Problem and a Motivating Experiment}
\label{sec:prob}

Let $w_t \in \mathbb{R}$, $t = 1\dots,N$ denote the elements in a 1D array of length $N$. We sometimes refer to $w_t$ as a \emph{weight}. The maximum (contiguous) subarray problem is to find an interval $I = [m, \dots, M]$, $1 \leq m \leq M \leq N$, with maximal sum: 
\begin{equation}
\label{eqn:plain_kadane}
    w^* = \max_{m, M} \sum_{t=m}^{M} w_t  ~~~ (P1)
\end{equation}
This problem can be solved in $O(N)$ time\footnote{Certain generalizations of this problem with constraints and maximum-density objective also allow an $O(N)$ solution \cite{max_density_segment_Lu}.} by an elegant dynamic-programming formulation proposed by Kadane \cite{bentley84}. 

\begin{figure}[t]
\centerline{\includegraphics[width=7.8cm]{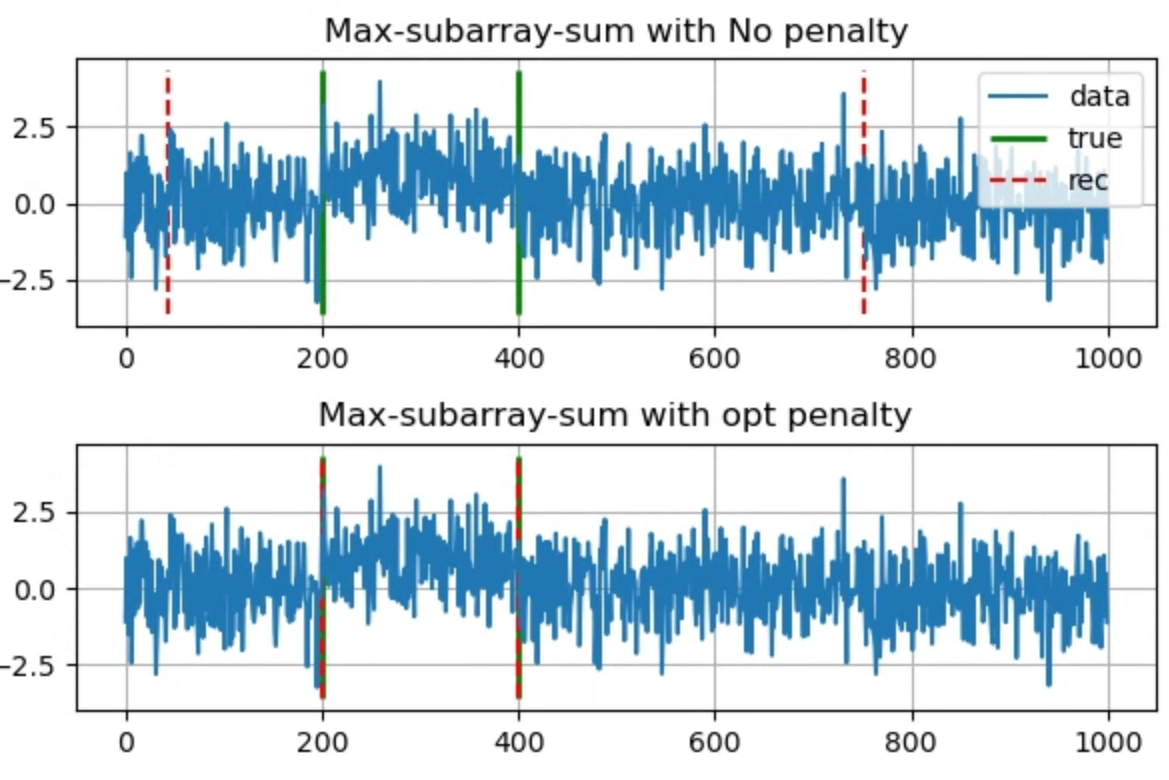}}
\caption{\label{fig:kadane_1d_no_penalty} Max-subarray localization of an interval with elevated mean. Green vertical lines show the true interval, and dashed red lines the estimated interval. (top) no penalty $\delta$ (bottom) mean/2 penalty. }
\vspace{-.1in}
\end{figure}

In this paper, we consider the use of maximum subarray \eqref{eqn:plain_kadane} and its extensions for a statistical problem of localizing the interval $[m, \dots, M]$ in high noise when it has a mean different from the rest of the array. Let us consider a motivating experiment to gain intuition. We take $w_t$ initially to be i.i.d.~zero-mean Gaussian random variables $\mathcal{N}(0,1)$, and then increase the mean by $\mu>0$ in the interval $I$. Naively, one can attempt to recover the interval by solving \eqref{eqn:plain_kadane}. 
However, even with $\mu$ large enough to visibly stand out from the noise, the recovered interval is grossly inaccurate, see Figure \ref{fig:kadane_1d_no_penalty} (top).  The reason for these large errors is the long expected run-length of cumulative sums of i.i.d.~random variables that gives rise to large false-positive regions, see Figure \ref{fig:kadane_run_lengths} (top). In Section~\ref{sec:alt_loc} we show that 
the errors in the recovered interval have expected length on the order of the full array length, making accurate localization impossible.

The failure of \eqref{eqn:plain_kadane} to localize can be rectified in two ways. First, we can subtract a penalty $\delta$ from each weight, yielding 
\begin{equation}
    \label{eqn:penalized_kadane}
    w^*(\delta) = \max_{m, M} \sum_{t=m}^{M} (w_t - \delta)~~~~~~ (P2)
\end{equation}
Alternatively, we can impose an upper bound $K$ on the interval length $|I| = M-m+1$:
\begin{equation}\label{eqn:constrained_kadane}
    w^*(K) = \max_{m, M} \sum_{t=m}^{M} w_t, ~~ M - m + 1 \le K~~~~ (P3)
\end{equation}
With an appropriate choice of $\delta$ in \eqref{eqn:penalized_kadane}, the errors in Figure~\ref{fig:kadane_1d_no_penalty} (top) reduce dramatically, accurately localizing the interval, Figure \ref{fig:kadane_1d_no_penalty} (bottom). 
In Section~\ref{sec:stat_form}, we present a statistical derivation of \eqref{eqn:penalized_kadane}, which specifies an ideal value for $\delta$. The relationship between the penalized and constrained versions \eqref{eqn:penalized_kadane} and \eqref{eqn:constrained_kadane} is discussed in Section~\ref{sec:lagrange}.

\section{Statistical Formulation}
\label{sec:stat_form}

In this section, we show that the penalized maximum subarray problem \eqref{eqn:penalized_kadane} can be derived from a statistical localization problem with exponential families. 

We consider $w_1,\dots,w_N$ to be independent random variables. Each $w_t$ follows an exponential family distribution with canonical-form density
\begin{equation}\label{eqn:exp_family}
    f(w_t) = h(w_t) \exp\left(\eta w_t + \eta'^T T(w_t) - A(\eta, \eta') \right),
\end{equation}
where one of the sufficient statistics is $w_t$ itself with corresponding natural parameter $\eta$. The (possibly vector-valued) function $T(w_t)$ captures any other sufficient statistics with natural parameters $\eta'$, and $A(\eta, \eta')$ is the log-partition function. 

We assume that for $t$ belonging to an unknown interval $[m,\dots,M]$, 
we have $\eta = \eta_1$, and elsewhere $\eta = \eta_0 < \eta_1$. The joint probability density of $(w_1,\dots,w_N) = w$ is therefore 
\begin{multline}\label{eqn:pdf}
    f(w) = \prod_{t=m}^{M} h(w_t) \exp\left(\eta_1 w_t + \eta'^T T(w_t) - A(\eta_1, \eta') \right)\\
    \times \prod_{t\notin\{m,\dots,M\}} h(w_t) \exp\left(\eta_0 w_t + \eta'^T T(w_t) - A(\eta_0, \eta') \right).
\end{multline}
We regard $\eta_0$, $\eta_1$, and $\eta'$ as known parameters and localize the interval by estimating $m, M$ via maximum likelihood. After taking the logarithm, \eqref{eqn:pdf} can be rewritten as the following log-likelihood:
\begin{multline}\label{eqn:log-likelihood}
    \ell(m, M \given w) = \sum_{t=m}^{M} \left((\eta_1 - \eta_0) w_t - A(\eta_1, \eta') + A(\eta_0, \eta') \right)\\
    {} + \sum_{t=1}^N \left( \log h(w_t) + \eta_0 w_t + \eta'^T T(w_t) - A(\eta_0, \eta') \right).
\end{multline}
The second line in \eqref{eqn:log-likelihood} does not depend on $m, M$ and can thus be omitted from the maximization. The remaining quantity is then proportional to the objective function in \eqref{eqn:penalized_kadane} with
\begin{equation}\label{eqn:delta}
    \delta = \frac{A(\eta_1, \eta') - A(\eta_0, \eta')}{\eta_1 - \eta_0}.
\end{equation}

The localization problem of maximizing log-likelihood \eqref{eqn:log-likelihood} can also be embedded into a binary hypothesis test of whether there exists an interval with parameter $\eta_1 \neq \eta_0$. Here the null hypothesis $H_0$ is that $\eta = \eta_0$ for all $t$ while the alternative $H_1$ is as described above with joint density \eqref{eqn:pdf}. Since $m$, $M$ are unknown, $H_1$ is a composite hypothesis. We can use a generalized likelihood ratio test (GLRT) to replace $m$, $M$ with their maximum likelihood estimates, by maximizing \eqref{eqn:log-likelihood}. The resulting maximal value of the first line in \eqref{eqn:log-likelihood} is then the generalized log-likelihood ratio, to be compared to a threshold to decide between $H_0$ and $H_1$.

Expression \eqref{eqn:delta} specifies the penalty $\delta$ in terms of the parameters $\eta_0$, $\eta_1$, $\eta'$. It can be interpreted as follows. First, it is the threshold to apply to $w_t$, $t \in [m,\dots,M]$, that results from the likelihood ratio test comparing $f(w_t \given H_1) / f(w_t \given H_0)$ to $1$ (i.e., equal priors on the hypotheses). Here the $t$th term in the first line of \eqref{eqn:log-likelihood} is the log-ratio $\log (f(w_t \given H_1) / f(w_t \given H_0))$. 
Second, this threshold \eqref{eqn:delta} always lies between the mean values of $w_t$ under $H_0$ and $H_1$, $\mu_0 \coloneqq \EE{w_t \given H_0}$ and $\mu_1 \coloneqq \EE{w_t \given H_1}$.
\begin{prop}
For an exponential family distribution as in \eqref{eqn:exp_family}, the threshold $\delta$ in \eqref{eqn:delta} satisfies $\mu_0 \leq \delta \leq \mu_1$.
\end{prop}
\begin{proof}
We exploit two properties of the log-partition function $A(\eta, \eta')$ of exponential families. First, the partial derivatives of $A(\eta, \eta')$ give moments of the corresponding sufficient statistics, and in particular, 
\begin{equation}\label{eqn:dA/deta}
    \left.\frac{\partial A}{\partial\eta}\right\rvert_{\eta=\eta_i} = \EE{w_t \given H_i} = \mu_i, \quad i =0, 1.
\end{equation}
Second, $A(\eta, \eta')$ is a convex function. Viewing $A(\eta, \eta')$ as a function of $\eta$ for fixed $\eta'$, the value of $\delta$ in \eqref{eqn:delta} can be recognized as the slope of the chord connecting the points $(\eta_0, A(\eta_0,\eta'))$ and $(\eta_1, A(\eta_1,\eta'))$. By convexity, the slope of this chord must be in between the slopes of the tangents at $\eta_0$, $\eta_1$, namely $\mu_0$, $\mu_1$ \eqref{eqn:dA/deta}. Convexity also implies that $\partial A/\partial\eta$ is a non-decreasing function of $\eta$, and with the assumption $\eta_0 < \eta_1$, this gives the ordering $\mu_0 \leq \delta \leq \mu_1$ ($\eta_0 > \eta_1$ would be analogous).
\end{proof}

\noindent\textbf{Example: Gaussian.} The Gaussian distribution $\mathcal{N}(\mu, \sigma^2)$ is an exponential family of the form in \eqref{eqn:exp_family} with $\eta = \mu/\sigma^2$, $\eta' = -1/(2\sigma^2)$, and $T(w_t) = w_t^2$. The log-partition function is given by 
\[
A(\eta, \eta') = -\frac{\eta^2}{4\eta'} - \frac{1}{2} \log(-2\eta') = \frac{\mu^2}{2\sigma^2} + \frac{1}{2} \log(\sigma^2).
\]
Substitution into \eqref{eqn:delta} yields $\delta = \frac{\mu_0 + \mu_1}{2}$ as the optimal penalty, i.e., the midpoint between the means.

\noindent\textbf{Example: Poisson.} The Poisson distribution with rate $\lambda$ satisfies \eqref{eqn:exp_family} with $\eta = \log\lambda$, no $\eta'$ or $T(w_t)$, and $A(\eta) = e^\eta = \lambda$. The optimal penalty \eqref{eqn:delta} is thus given by 
\begin{equation}\label{eqn:deltaPoisson}
    \delta_{\mathrm{Pois}} = \frac{\lambda_1 - \lambda_0}{\log\lambda_1 - \log\lambda_0}.
\end{equation}
It can be verified that this is always between $\lambda_0$ and $\lambda_1$.

In practice it is rare to know both background and foreground means $\mu_0$ and $\mu_1$. A more typical case is that we know or can estimate the background mean $\mu_0$ (often it can be assumed to be $0$), and we have a prior on the difference $\Delta\mu = \mu_1 - \mu_0$, or have a tolerance in mind, where deviations below $\Delta \mu$ may not be of interest. Then, we can set $\mu_1=\mu_0 + \Delta \mu$, and proceed as if both were known. 

The above derivation has close parallels to offline change-point detection \cite{truong2020selective} and more distantly to CUSUM statistics for online change-point detection \cite{CUSUM_Page_1954,Granjon2014CUSUM,basseville1993detection}, specifically in the use of maximum likelihood and exponential families \cite{frick2014multiscale}. 
Our interval localization problem can be seen as an offline problem of detecting two changes, but where the distributions before and after the interval are the same. This property could be 
why it admits the more efficient $O(N)$ algorithm of Kadane rather than the $O(N^2)$ of the more generic dynamic programming algorithm \texttt{Opt} in \cite{truong2020selective} for $K=2$.

\section{Localization Analysis}
\label{sec:alt_loc}

\begin{figure}[t]
\centerline{\includegraphics[width=8cm]{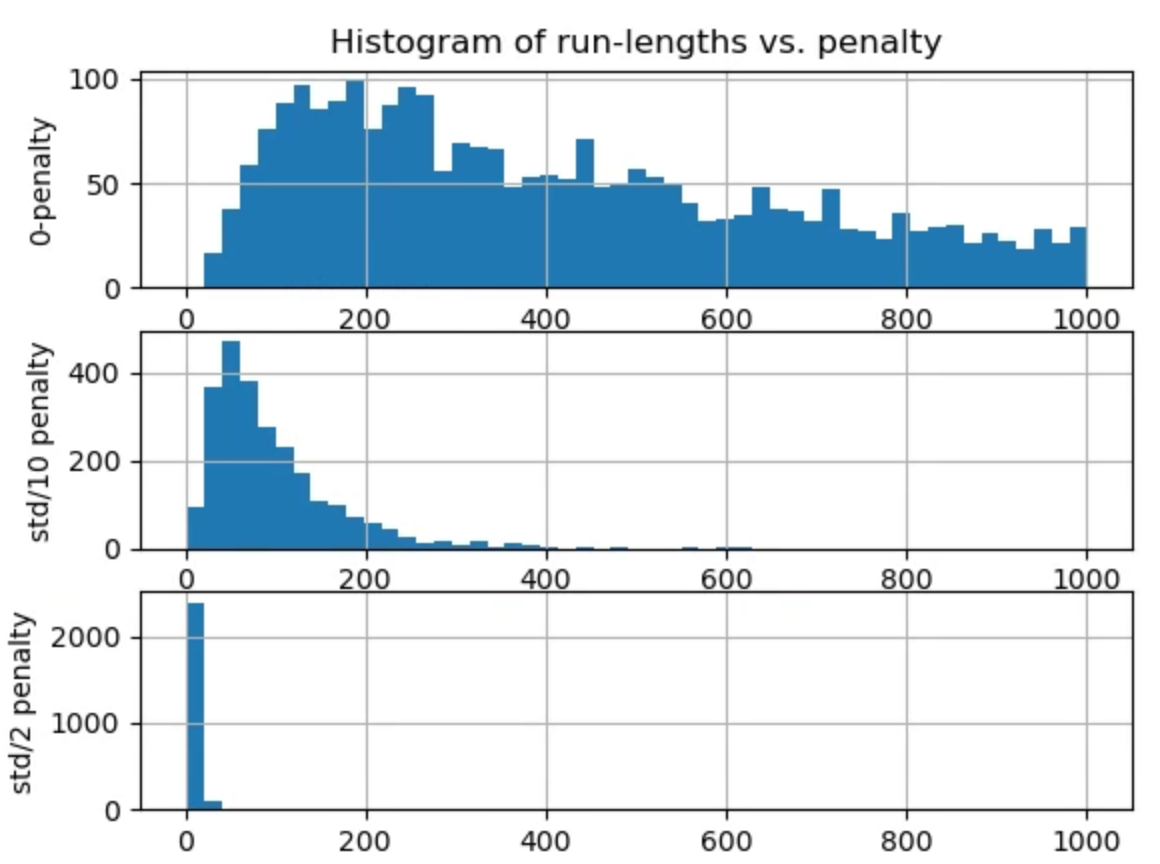}}
\caption{ \label{fig:kadane_run_lengths} Histogram of run-lengths (i.e., lengths) of max subarray-sum regions for zero-mean Gaussian weights, $N=1000$. (top) penalty $\delta = 0$ (middle) $\delta = 1/10$ (bottom) $\delta = 1/2$. 
}
\end{figure}

In this section we analyze the localization error of the penalized max-subarray sum algorithm (P2) in (\ref{eqn:penalized_kadane}). Denote by $\hat{m}, \hat{M}$ the estimated boundaries corresponding to the $\argmax$ in (P2), and $m, M$ the true boundaries. $N$ is the length of the full array.  We will show that for $\delta=0$, accurate localization is essentially impossible with the error $\hat{M}-M \propto N-M$ in the case $\Mh \geq M$, i.e., the localization error is proportional to $N$ (we focus on the right boundary $M$, but the same analysis applies to left boundary $m$). On the other hand, as soon as $\delta > 0$, the localization error is 
independent of $N$. While precise distributions of run-lengths and localization errors can be derived based on stopping times of discrete Brownian motion (Gaussian random walks) with reflecting and absorbing boundaries \cite{domine1996firstPassage, blasi1976_random_walk}, here we take a shortcut using symmetry arguments. The authors of \cite{drawdown_magdon_ismail_2004} analyze the maximum drawdown problem (equivalent to the maximum subarray after changing array sign), and characterize the weights of the max subarray regions. However, they do not address the lengths of the max subarrays. 

\begin{lemma}
In the i.i.d.~Gaussian setting, with $\hat{M} \geq M$, the expected localization error of (\ref{eqn:penalized_kadane}) with $\delta=0$ is $E[\hat{M}-M] = (N-M)/2$. 
\end{lemma}
\begin{proof}[Proof sketch] Given $\Mh \geq M$, consider the cumulative sum of the tail portion of the array beyond $M$: $c(\tau) = \sum_{t=M+1}^{\tau} w_t$. If the region $M+1, ..., \tau$ has a positive weight $w_{M+1} + \dots + w_\tau > 0$, then subarray $[m, ..., \tau]$ would be an improvement over $[m, M]$, and hence we will overestimate the true region ending at $M$.  The optimal value of the improvement is achieved at $\arg\max_{M < t \le N} c(t)$.
Now $c(t)$ from $M+1$ to $N$ is a simple Gaussian random walk, a sum of i.i.d.~Gaussian increments with $0$-mean. So for any trajectory $w_{M+1}, ... , w_N$  which achieves a maximum in $c(t)$ at $M+1+ \tau$, there is a reverse trajectory $v_{M + 1 + i} = w_{N-i}$, which achieves a maximum at $N - \tau$. So by pairing each trajectory with its reverse, the expected location of the $\arg\max$ of $c(t)$ is exactly the middle $M + (N - M)/2$, and the expected error is $(N-M)/2$.
\end{proof}

For the converse, as soon as we use any penalty $\delta > 0$, our localization error stops being a function of $N$. 
\begin{lemma}
For $\delta>0$ the expected localization error is independent of the length of the array $N$.
\end{lemma}
\begin{proof}[Proof sketch] Consider again $c(t)$ for $t > M$.  It can be decomposed into a negative linear trend term $-\delta (\tau - M)$ and a standard zero-mean Brownian motion $\bar{c}(t)$. From the analysis in \cite{drawdown_magdon_ismail_2004}, the weight of the maximum region in the zero-mean array $\bar{c}(t)$ over $[M, ... ,\tau]$ is $O(\sqrt{\tau-M})$. For $\tau$ large enough (which is a function of the constant in front of $\sqrt{\tau - M}$, but not of $N$), the linear trend term will dominate and the sum will be negative. Therefore the optimal $\hat{\tau}$ occurs earlier than this and $\hat{\tau}-M$ does not depend on $N$. 
\end{proof}

\section{Lagrangean Interpretation}
\label{sec:lagrange}

\begin{figure}[t]
\centerline{\includegraphics[width=7.5cm]{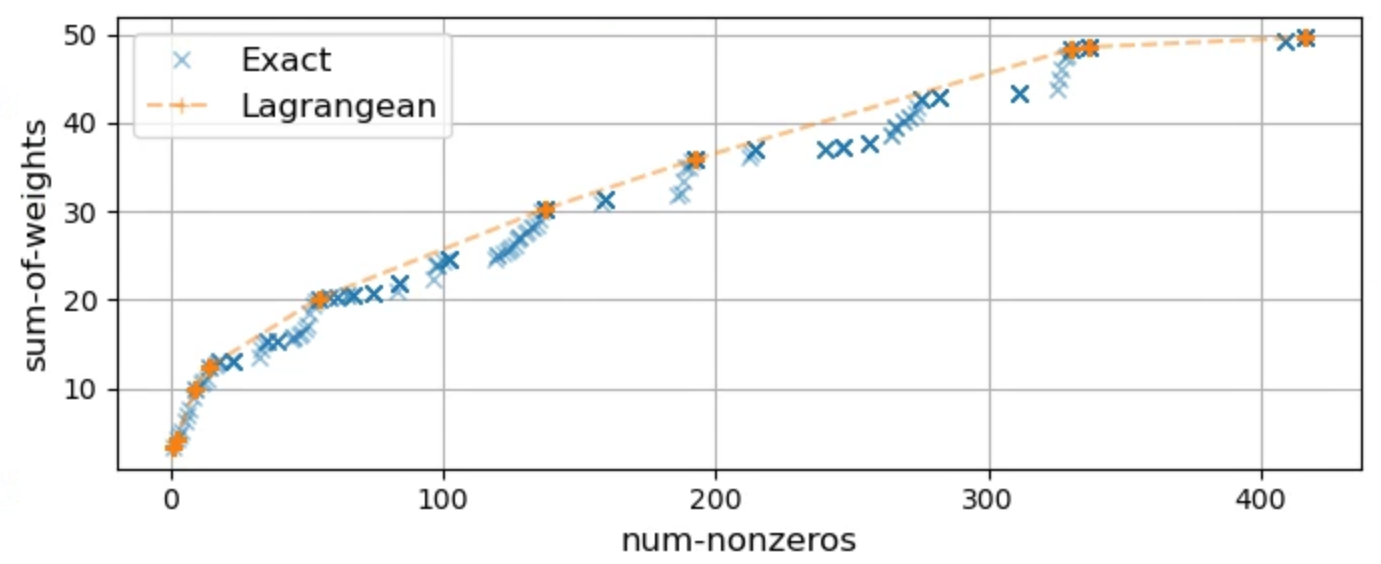}}
\caption{ \label{fig:lagrangean_convex_hull} Lagrangean relaxation (P2, orange) finds a subset of length-constrained solutions (P3, blue) that lie on vertices of the convex hull. For nonzeros $\in [416, 500]$, the solution to (P3) is the same.}
\end{figure}

We now relate the length-constrained max subarray sum in \eqref{eqn:constrained_kadane} to the statistical formulation in Section \ref{sec:stat_form}. In principle, if we want to find a solution with a desired cardinality, there exist efficient $O(N)$ algorithms that directly solve the length-constrained max subarray problem \cite{lin2002efficient,max_density_segment_Lu}. Here, our goal is to understand the relationship between the two problems. 

Adding a Lagrangean penalty on the region-size constraint in \eqref{eqn:constrained_kadane} and simplifying, we obtain the penalized version in \eqref{eqn:penalized_kadane}. This suggests two possibilities for using \eqref{eqn:penalized_kadane}. First, based on the discussion in Section \ref{sec:stat_form}, we can apply \eqref{eqn:penalized_kadane} directly if we have a prior on $\Delta \mu$, by setting $\mu_1 = \mu_0 + \Delta\mu$ and 
$\delta$ accordingly. Alternatively, we can use \eqref{eqn:penalized_kadane} to find a solution satisfying the length constraint in \eqref{eqn:constrained_kadane} via a bisection search to find the smallest $\delta$ that satisfies $M - m + 1 \le K$. Within each bisection iteration we have to solve the unconstrained problem (P1) with modified weights $\tilde{w}_t = w_t - \delta$. What is the relationship between the families of solutions to \eqref{eqn:penalized_kadane} and \eqref{eqn:constrained_kadane}?

Unlike convex optimization problems (where under some technical conditions strong duality holds), for these discrete optimizations, the set of solutions $\{w^*(K)\}$ to (P3) is not equivalent to $\{w^*(\delta)\}$ of (P2). We observe empirically that $\{w^*(\delta)\}$ includes those solutions from $\{w^*(K)\}$ that lie at the corner points of the convex-hull of all solutions in the (weight, cardinality) space. We show an example in Figure \ref{fig:lagrangean_convex_hull}. We have $N=500$, the optimal subarray with $\delta=0$ has length 416, $\{w^*(K)\}$ (labeled exact) is shown with $K=[1,\dots,500]$ in steps of 1. We see that the set of solutions $\{w^*(\delta)\}$ of (P2), in orange, are the corner-points on the convex-hull of (P3) solutions $\{w^*(K)\}$, in blue. One could argue that these penalized solutions on the convex hull have a particularly good trade-off of region-weight vs cardinality w.r.t. other constrained solutions.

We do not formally prove this result here, but we conjecture that it comes from integer programming duality, namely the equivalence of the Lagrangean dual to the linear program (LP) over the convex hull for the remaining constraints (without the dualized constraints), see \cite[Thm.~11.4]{LP_book_bertsimas1997}.  Max subarray can be represented as an integer linear program, with binary variables $z_t$ denoting whether $t$ belongs to the region. By dualizing the region-size constraint $\sum z_t \le K$, we are left with the LP over the convex hull of $z$ satisfying contiguity constraints. We leave a formal proof for future work.

\begin{figure}[t]
\centerline{\includegraphics[width=8cm]{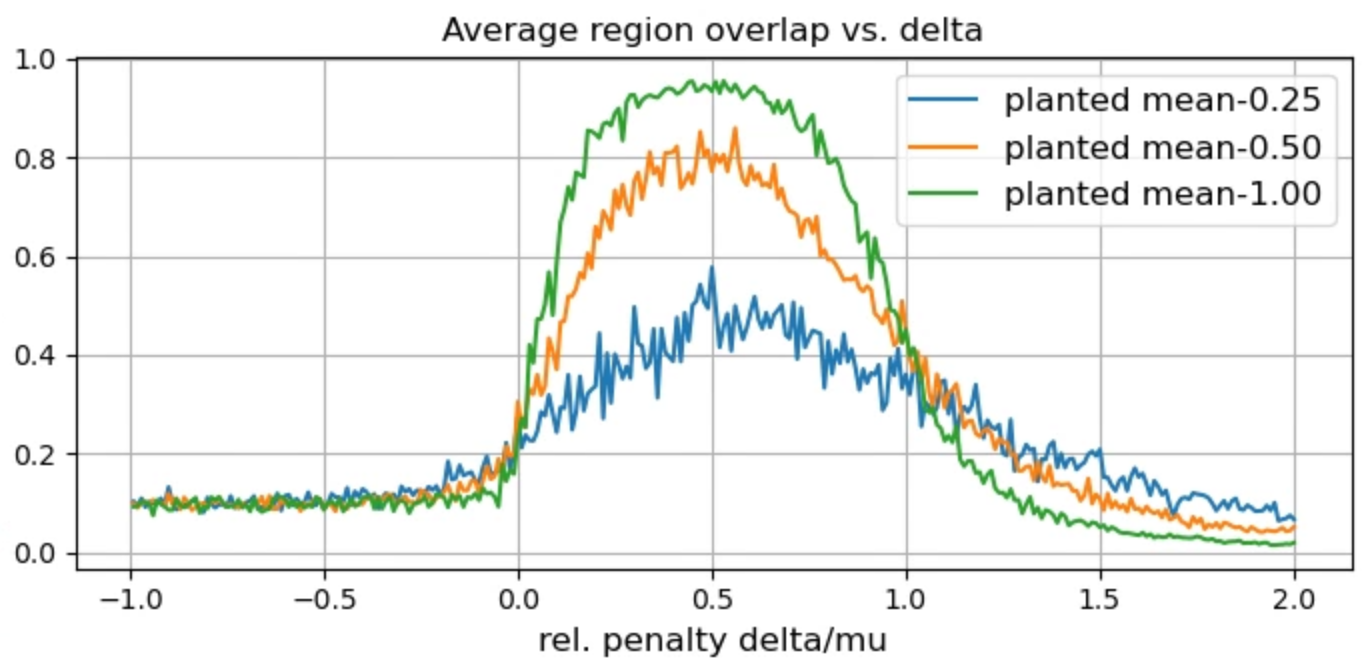}}
\centerline{\includegraphics[width=8cm]{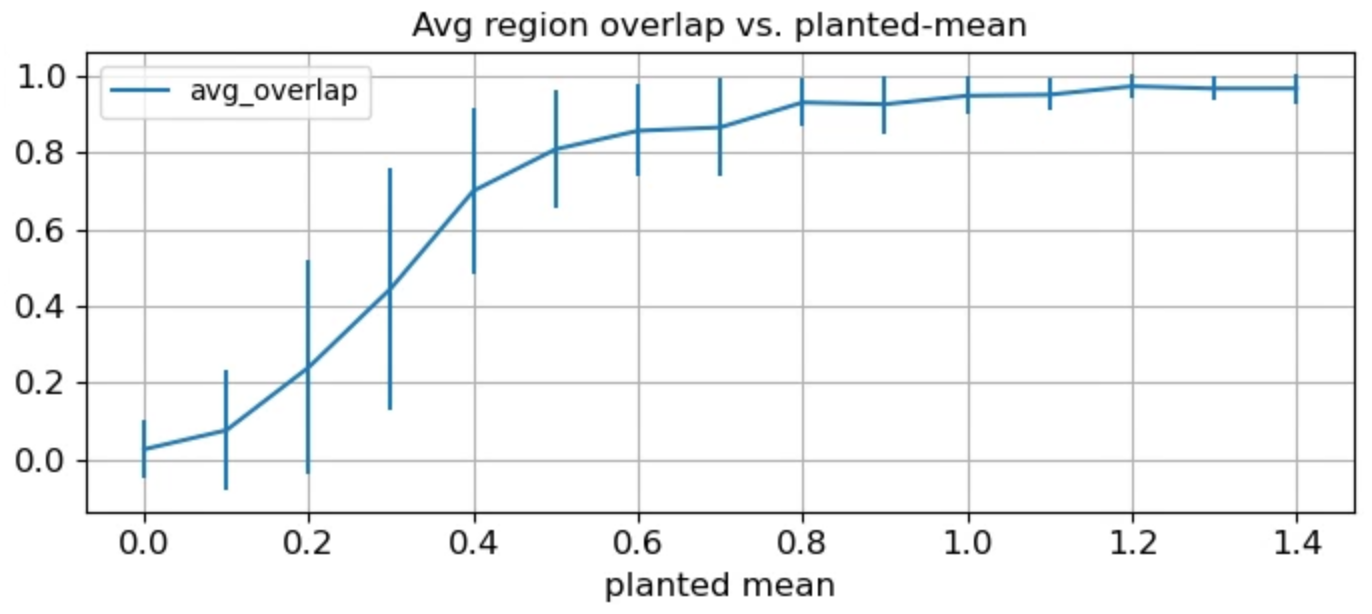}}
\caption{ \label{fig:kadane_region_loc} Region localization for Gaussian random variables. (top) Region overlap vs.~$\delta / \mu$. (bottom) vs.~$\mu$ for optimal $\delta$.}
\vspace{-0.1in}
\end{figure}

\section{Numerical Simulations}
\label{sec:expt}

We illustrate the performance of penalized max subarray-sum in planted region localization in Figure \ref{fig:kadane_region_loc}.  We take a zero-mean unit-variance i.i.d.~Gaussian vector of length $N=1000$, and plant a region of size $100$ with mean $\mu$. We compute
the overlap of the ground-truth region $I^*$ with the estimated one $\hat{I}$ based on the overlap metric $\frac{|\hat{I} \cap I^*|}{|\hat{I} \cup I^*|}$. In the upper plot we show
overlap as a function of $\delta / \mu$ for 3 values of $\mu$, $[0.25, 0.5, 1.0]$. We see that
with a heavy mean $\mu=1$ there is near-perfect localization for $\delta=\mu/2$ and near-optimal performance for a range of $\delta$,  while for situations with weaker signal, overlap is still maximized at $\delta=\mu/2$, but the maximum overlap value is reduced.  In the lower plot we instead vary $\mu$, and set $\delta=\max(\mu/2, 0.25)$. We set a floor of $0.25$ as below it we would effectively run unconstrained max subarray-sum, with overly long 
recovered regions. 

In Figure \ref{fig:kadane_Poisson_loc} (top) we illustrate localization in the Poisson setting. The background $w_t$'s have rate $\lambda_0=5$, while in a planted region the rate is slightly higher at $\lambda_1=6$.  The optimal Poisson penalty $\delta_{\mathrm{Pois}}$ \eqref{eqn:deltaPoisson} allows accurate localization despite the region rate being close to background. In the lower plot we show region overlap versus $\delta$. Performance decays away from $\delta_{\mathrm{Pois}}$.

\begin{figure}[t]
\centerline{\includegraphics[width=7.9cm]{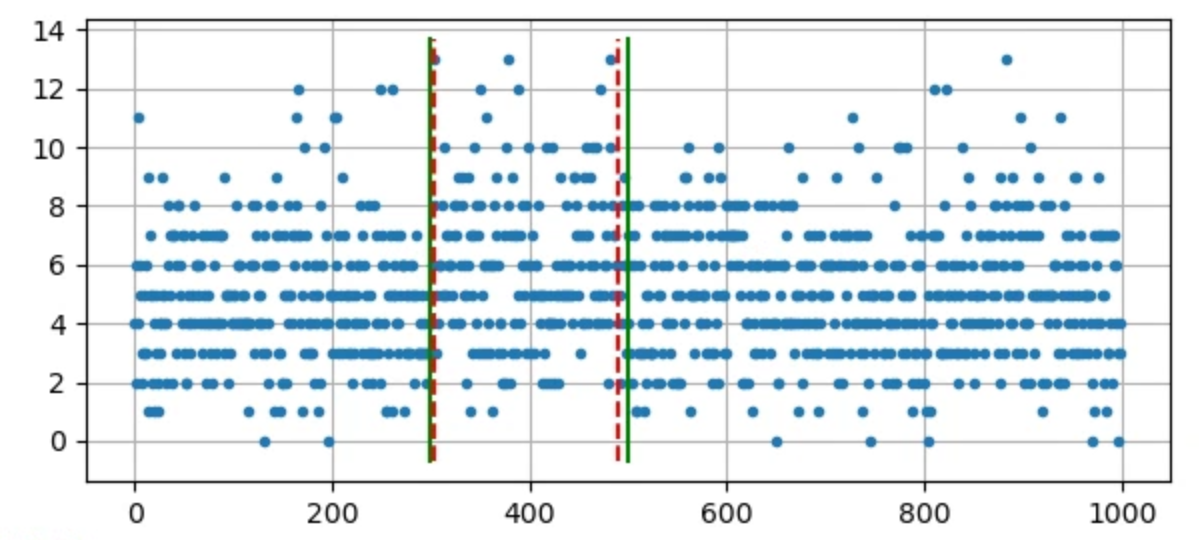}}
\centerline{\includegraphics[width=7.9cm]{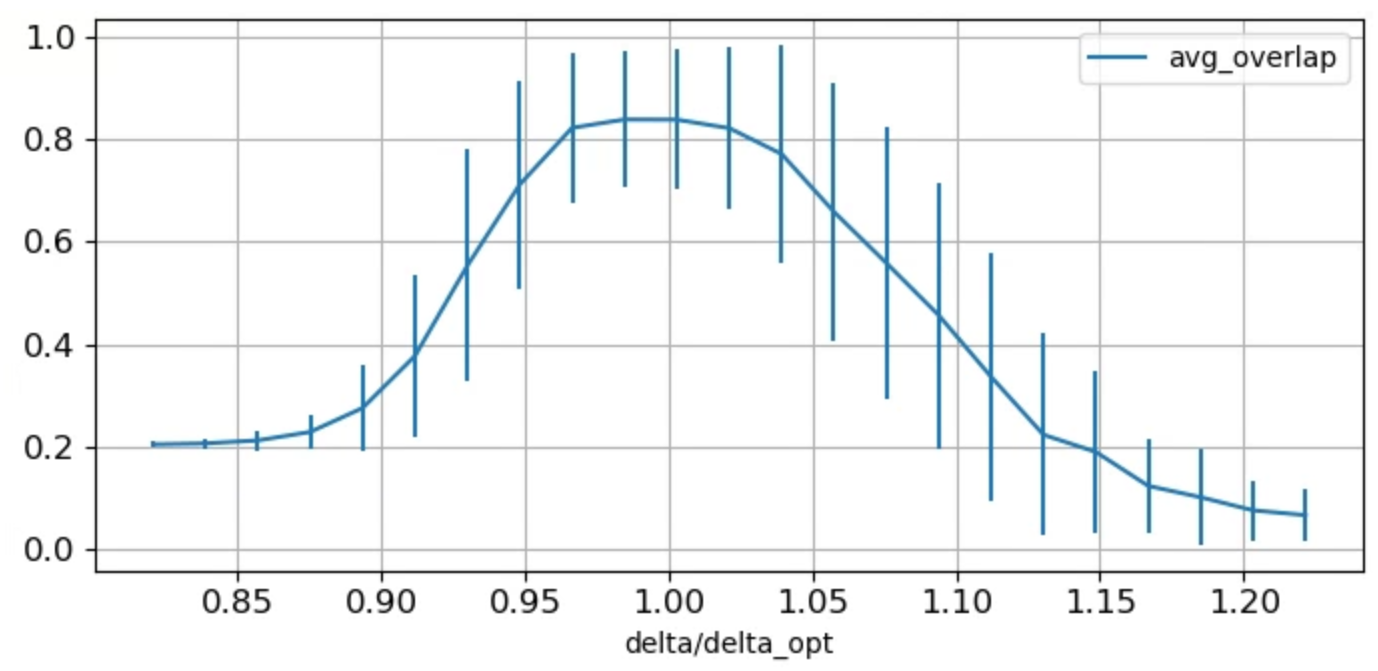}}
\caption{ \label{fig:kadane_Poisson_loc} Region localization for Poisson random variables. (top) Example time-series and localization under optimal $\delta = \delta_{\mathrm{Pois}}$. (bottom) Region overlap 
as a function of $\delta/\delta_{\mathrm{Pois}}$ (`delta\_opt').}
\vspace{-0.1in}
\end{figure}

In Figure \ref{fig:SAR} we apply a 2D version of max subarray to 
vehicle localization in a SAR image.\footnote{\url{https://www.sandia.gov/app/uploads/sites/124/2021/06/Ka-band-image-of-a-variety-of-military-vehicles-in-the-desert-near-Albuquerque-NM..png}, courtesy of Sandia National Laboratories, Radar ISR} We set $\delta$ using a prior on expected 
object size (corresponding to $K$) 
and apply max-2D-subarray iteratively, masking objects found earlier. Varying $\delta$ by $\pm 15\%$ results in similar detected regions, as seen in Figure~\ref{fig:SAR_sensitivity}. In this illustrative example the analysis is done on raw pixels, without any preprocessing (e.g.~edge filters, etc.).

\begin{figure}%
    \centering
    \subfloat[\centering]{{\includegraphics[width=3.8cm]{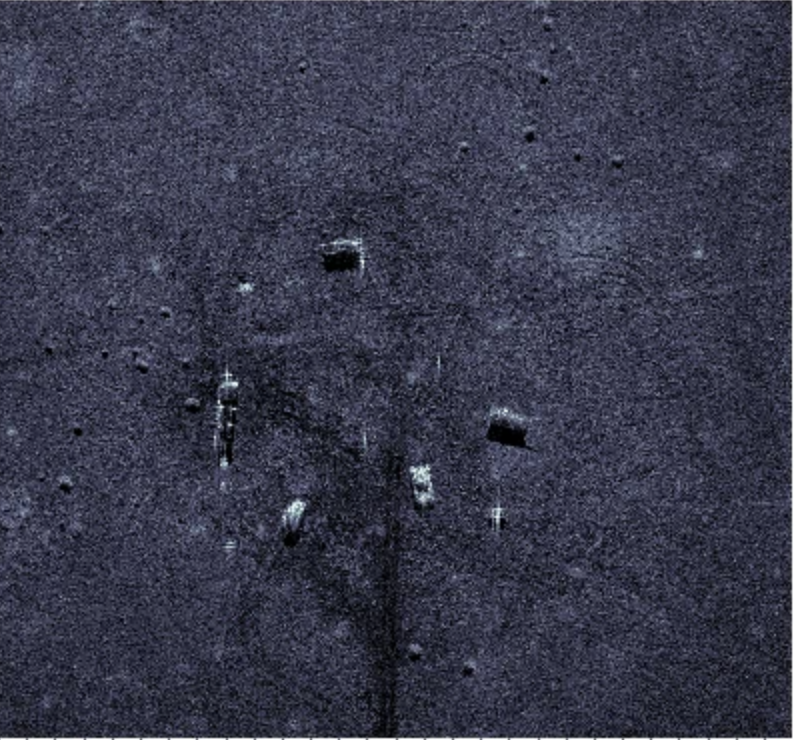} }}%
    \qquad
    \subfloat[\centering]{{\includegraphics[width=3.8cm]{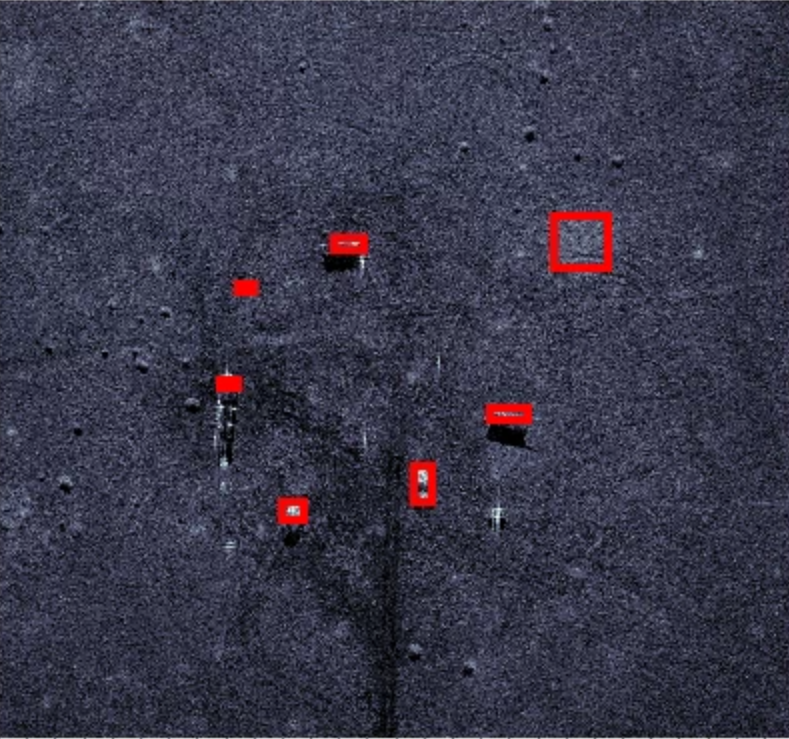} }}%
    \caption{ \label{fig:SAR} SAR object localization. (a) image (b) max subarray regions, estimated iteratively by masking regions found thus far.}
    \vspace{-0.1in}
\end{figure}

\begin{figure}[h]%
    \centering
    \subfloat[\centering]{{\includegraphics[width=3.9cm]{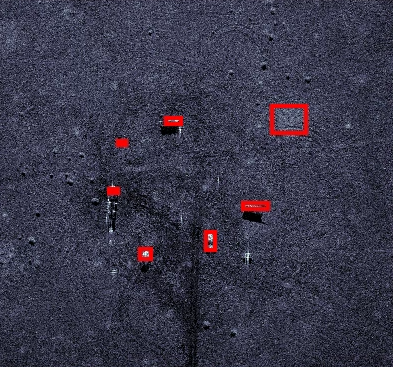} }}%
    \qquad
    \subfloat[\centering]{{\includegraphics[width=3.9cm]{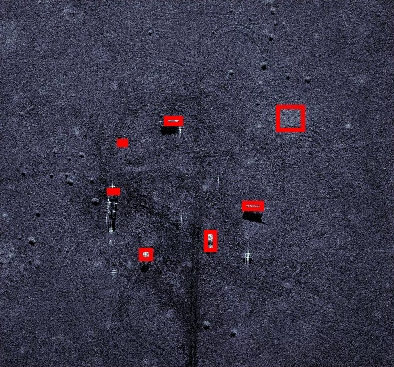} }}%
    \caption{ \label{fig:SAR_sensitivity} SAR object localization using maximum subarray with $\delta$ changed by (a) $-15\%$ and (b) $+15\%$ from the value in Fig.~\ref{fig:SAR}(b).}
\end{figure}

\vfill\pagebreak

\bibliographystyle{IEEEbib}
\bibliography{heavy_sets}

\clearpage

\end{document}

%% file: ICASSP/definitions.tex
\newcommand{\given}[1][]{\,#1\vert\,}

\newcommand{\EE}[1]{\mathbb{E}\left[#1\right]}
\DeclareMathOperator*{\argmax}{arg\,max}

\newtheorem{prop}{Proposition}
\newtheorem{lemma}{Lemma}

\newcommand{\Mh}{\hat{M}}